\documentclass[12pt]{article}

\usepackage[top=1in, left=1in, right=1in,
bottom=1in]{geometry}   
\usepackage[parfill]{parskip}
\usepackage{graphicx}
\usepackage{amssymb}
\usepackage{epstopdf}
\usepackage{bbm}
\usepackage{amssymb}
\usepackage{amsmath}
\usepackage{amsfonts}

\usepackage{lscape}
\usepackage{rotating}
\usepackage{setspace}
\usepackage{threeparttable}
\usepackage{booktabs}
\usepackage[compact]{titlesec}

\usepackage{lmodern}
\fontfamily{lmtt}\selectfont
\usepackage[T1]{fontenc}

\usepackage{natbib}
\bibpunct{(}{)}{;}{a}{}{,}

\usepackage{hyperref}

\usepackage{titling}

\usepackage{amsthm}
\newtheorem{theorem}{Theorem}[section]

\newtheorem{proposition}[theorem]{Proposition}

\newtheorem{problem}[theorem]{Problem}
\newtheorem{algorithm}[theorem]{Algorithm}

\usepackage{color}
\usepackage{adjustbox}
\usepackage{multirow}

\begin{document}
\pagestyle{plain}

\newcommand{\blind}{0}
\newcommand{\tit}{\Large Optimal Multilevel Matching in Clustered Observational Studies: \\ \textcolor{black}{A Case Study of the Effectiveness of Private Schools \\ Under a Large-Scale Voucher System}}

\if0\blind
{\title{\tit\thanks{For comments and suggestions, we thank three anonymous reviewers, an associate editor, and Joseph Ibrahim.  We also thank Magdalena Bennett, \textcolor{black}{Jake Bowers}, \textcolor{black}{Nicol\'{a}s Grau}, Cinar Kilcioglu, \textcolor{black}{Winston Lin}, Sam Pimentel and Paul Rosenbaum, and seminar participants at Johns Hopkins University and the University of Pennsylvania.  This work was supported by a grant from the Alfred P. Sloan Foundation.} }
\author{Jos\'{e} R. Zubizarreta\thanks{Assistant Professor, Division of Decision, Risk and Operations, and Statistics Department,  Columbia University, 3022 Broadway, 417 Uris Hall, New York, NY 10027, Email: zubizarreta@columbia.edu.}
\and Luke Keele\thanks{Associate Professor, Department of Political Science, 211 Pond Lab, Penn State University, University Park, PA 16802 Phone: 814-863-1592, Email: ljk20@psu.edu.}
}
\date{}
\maketitle
}\fi


\begin{abstract}
A distinctive feature of a clustered observational study is its multilevel or nested data structure arising from the assignment of treatment, in a non-random manner, to groups or clusters of units or individuals. Examples are ubiquitous in the health and social sciences including patients in hospitals, employees in firms, and students in schools. What is the optimal matching strategy in a clustered observational study? At first thought, one might start by matching clusters of individuals and then, within matched clusters, continue by matching individuals. But as we discuss in this paper, the optimal strategy is the opposite: in typical applications, where the intracluster correlation is not perfect, it is best to first match individuals and, once all possible combinations of matched individuals are known, then match clusters. In this paper we use dynamic and integer programming to implement this strategy and extend optimal matching methods to hierarchical and multilevel settings. Among other matched designs, our strategy can approximate a paired clustered randomized study by finding the largest sample of matched pairs of treated and control individuals within matched pairs of treated and control clusters that is balanced according to specifications given by the investigator. This strategy directly balances covariates both at the cluster and individual levels and does not require estimating the propensity score, although the propensity score can be balanced as an additional covariate. We illustrate our results with a case study of the comparative effectiveness of public versus private voucher schools in Chile, a question of intense policy debate in the country at the present.
\end{abstract}


\begin{center}
\end{center}
\clearpage
\doublespacing

\section{Introduction}

\subsection{Clustered Observational Studies with Multilevel Data}

Clustered observational studies are ubiquitous in the health and social sciences.  Examples include patients receiving similar treatments in hospitals, employees facing a policy change inside firms, and students following a particular learning program within schools.  Clustered observational studies have a nested or multilevel data structure with observed and unobserved covariates both at the cluster and unit levels. In this context, research interest typically lies in the effect of the cluster level treatment on unit level outcomes, however this effect may be confounded by differences in the distributions of covariates across the treatment groups both at the cluster and unit levels.  Therefore, an important question in clustered observational studies is: how to adjust for observed covariates taking into account the multilevel structure \textcolor{black}{of the data}?  \textcolor{black}{Ideally, these adjustments will balance covariates at the cluster and  unit levels in a transparent manner, while mimicking a target clustered randomized experiment and facilitating sensitivity analyses to hidden biases due to unobserved covariates \citep{Rosenbaum:2010}.}

\textcolor{black}{Perhaps the most well-known multilevel data structure arises in educational settings,} with unit level measures such as the student's score on a standardized test, and also cluster level covariates such as school enrollment \citep{Lee:1989}.  Both covariates may act as confounders when evaluating, for instance, the impact of a study program or administration regime targeted at improving learning. A conventional approach to adjust for cluster and unit level covariates is hierarchical or multilevel regression modeling. \textcolor{black}{Multilevel regression models can aid causal inference by accounting for the design of the data collection, and, under certain assumptions, by adjusting for unmeasured covariates and modeling variation in treatment effects \citep{feller_gelman_2015}.}

Nonparametric alternatives to multilevel regression modeling often rely on propensity scores \citep{Hong:2006,  Arpino:2011, Li:2013}. For example, \citet{Hong:2006} stratify on a multilevel propensity score to approximate a two-stage experiment where schools and students are randomly assigned to treatment within blocks. Matching methods are extensively used in observational studies \citep{stuart_2010, lu_et_al_2011}, however they do not typically account for multilevel data structures.

\textcolor{black}{
In this paper, we develop an optimal matching strategy for clustered observational studies. Contrary to intuition and common practice, which first matches clusters and then matches units within matched clusters, our strategy does the opposite: it first matches pairs of units across all possible combinations of treated and control clusters, and, once all the possible combinations of pairs of units across pairs of clusters are known, it then matches the clusters. As we discuss, our strategy is optimal in the sense that it minimizes the total sum of covariate distances between matched pairs of units (when implemented with optimal matching; \citealt{rosenbaum_1989a}), or that it maximizes size of the pair-matched sample that balances covariates as required by the investigator (when implemented with cardinality matching; \citealt{zubizarreta_et_al_2014}). With cardinality matching, the investigator can directly balance not only the first moments of the observed covariates but also their entire distributions. Cardinality matching does not require estimating the propensity score, although the propensity score can be used as an additional balancing covariate. Even though we implement our matching strategy in a clustered observational study with two levels ---students within schools--- it can readily be extended to settings with three or more levels such as students within schools within districts. In particular, we illustrate this method on a case study of the comparative effectiveness of public versus private voucher schools in standardized tests in Chile. This is a question of intense policy debate in the country at the present, and we describe it subsequently.
}
 
\subsection{Private and Public Schools Under a Voucher System in Chile}

\textcolor{black}{In Chile in the 1980's, the military dictatorship enacted a far-reaching educational reform as it implemented\textcolor{black}{, among others,} a universal voucher system. That system, which is still in force today, is based on a direct payment to each school as a function of daily attendance.  As \textcolor{black}{a} result of this reform, \textcolor{black}{a large number of private schools subsidized by the voucher emerged, and today students in Chile can} choose between three major types of schools. 
First, students can attend public schools which are run by local government authorities.  
Second, they can attend private subsidized schools. 
For both of these types of schools, the school receives a fixed amount of voucher money for each student 
\textcolor{black}{based on daily attendance.}
Private subsidized schools, however, tend to set admissions policies, design curriculum, and they can also charge additional tuition (copayment). 
These private subsidized schools are often referred to as voucher schools, even though public schools also operate under \textcolor{black}{a general} voucher system. 
Finally, there are private non-subsidized schools, which do not accept voucher payments from the government and are entirely funded by student tuition. 
\textcolor{black}{Today approximately 55\% of the students attend private subsidized schools, 37\% attend public schools, and the rest attend private non-subsidized schools \citep{mineduc2015a}}.}

\textcolor{black}{An important open question today is whether private subsidized schools deliver a better education than public schools. Following the student movement of 2011 in Chile \citep{economist2011, nytimes2011} a number of 
initiatives have been proposed (such as prohibiting the selection of students, co-financing, and eliminating profits for private subsidized schools; \citealt{bachelet2013, mineduc2015b}) which would likely result in closing many private subsidized schools 
and 
\textcolor{black}{replacing}
them with public schools. 
Current evidence on the effectiveness of private subsidized schools relative to public schools is mixed. A number of studies have found that private voucher schools increase test scores by at least 15\% to 20\% of a standard deviation \citep{Mizala:2001,Anand:2009} though other studies have found larger effects \citep{Sapelli:2002,Sapelli:2005}. 
\textcolor{black}{Other}
work has found effects that are either not statistically detectable \citep{Hsieh:2006,Mcewan:2001} or are much smaller \citep{Lara:2011}. }

\textcolor{black}{In this paper, we conduct an observational study of
\textcolor{black}{the comparative effectiveness of private subsidized or voucher schools versus public schools in terms of standardized test scores.}
\textcolor{black}{Our study seeks to understand the likely event of closure of entire private schools and replacement by public schools.}
As such, while students and their families decide to which schools to attend (within certain constraints), closing and replacing one type of school with another type of school is an intervention that happens at the school level and not at the student level.  Such reforms are not as common as interventions that seek to change one aspect of a school, but these reforms are consistent with a movement in education known as whole school or comprehensive reform.  Under whole school reform, reforms extend beyond curriculum to an overhaul of the entire school system \citep{edweek:2004}. One such example of whole school reform in the U.S. is Success for All (e.g., \citealt{borman2007}).}

\textcolor{black}{
\textcolor{black}{As we describe subsequently, our}
data have a multilevel structure in that we observe student level covariates such as gender and socio-economic status as well as school level covariates such as enrollment and whether the school is in an urban or rural area. 
Consistent with
\textcolor{black}{the hypothetical intervention as being applied to entire schools} as well as the data structure, we seek to mimic a clustered randomized experiment in which similar clusters are assigned to treatment and control, and furthermore, we also seek to remove overt bias to guarantee comparability at the student level. In this manuscript, we develop a matching strategy, consistent with both the data structure and a clustered treatment assignment mechanism, that pairs similar students within the paired schools.
}
 

\textcolor{black}{
This article is organized as follows. 
Section \ref{sec:not} presents the notation and basic assumptions, as well as the longitudinal census data and the design that we use in our study.  Section~\ref{sec:algo} reviews cardinality matching for finding the largest matched sample that is balanced, explains the multilevel matching strategy, and presents it in more generality. Section~\ref{sec:assessment} evaluates the matched sample. Section~\ref{sec:out} analyzes the comparative effectiveness of public and private voucher schools in Chile. Section~\ref{sec:conc} concludes with a summary and a discussion.
}

\section{\textcolor{black}{Notation, Data, and Study Design}}
\label{sec:not}

\subsection{Notation and Assumptions}

Here, we introduce the basic notation and assumptions. Before matching, we have a sample of treated and control students within treated and control schools. As we explain in more detail in Section \ref{sec:algo}, in our matched design we form pairs of treated and control students within pairs of treated and control schools. Following \citet{Hansen:2014}, after matching there are $K$ matched pairs of clusters, $k = 1, \dots, K$, with two schools, $j = 1, 2$, one treated and one control, so there are $2K$ units in total. Each cluster $kj$ contains $n_{kj} > 1$ individuals,  $i=1,\dots,n_{kj}$, each with a vector of observed covariates $\mathbf{x}_{kji}$. Similarly, $u_{kji}$ represents an unobserved covariate of individual unit $i$ in cluster ${kj}$. The covariates $(\mathbf{x}_{kji},u_{kji})$ may consist of measurements of either the student $kji$ or the school $kj$ of this student. In our study, we assume that treatment assignment occurs at the school level as whole schools are either private subsidized (voucher) or public schools. If the $j^{th}$ school in pair $k$ receives the treatment, we write $Z_{kj} = 1$, whereas if this school receives the control we write $Z_{kj} = 0$, so $Z_{k1} + Z_{k2} = 1$ for each $k$ as each pair contains one treated school and one control school.


Each student has two potential responses; one response that is observed under treatment $Z_{kj} = 1$ and the other observed under control $Z_{kj} = 0$ \citep{Neyman:1923a,Rubin:1974}. We denote these responses as $(y_{Tkji}, y_{Ckji})$, where $y_{Tkji}$ is observed from the $i$th subject in pair $k$ under $Z_{kj} = 1$, and $y_{Ckji}$ is observed from this subject under $Z_{kj} = 0$. In our application $y_{Tkji}$ is the test score that student $kji$ would exhibit if the student is enrolled in a private subsidized school and $y_{Ckji}$ is the test score this same student would exhibit if that school were converted to a public school. 
Of course, we do not observe both potential outcomes, but we do observe the responses $Y_{kji} = Z_{kj}y_{Tkji} + (1 - Z_{kj})y_{Ckji}$. Under this framework, the observed response $Y_{kji}$ varies with $Z_{kj}$ but the potential outcomes do not vary with treatment assignment.

Note that our notation implicitly assumes that there is no interference among units across schools. This assumption is often referred to as one part of the stable unit treatment value assumption \textcolor{black}{(SUTVA; \citealt{Rubin:1986})}. In our application, this assumption implies that potential test scores of a student in one school do not depend on the treatment assignment of students in other schools. In this context, $y_{Tkji}$ denotes the response of student $kji$ if all students in school $kj$ receive the treatment, while $y_{Ckji}$ denotes the response of student $kji$ if all students in school $kj$ receive the control. Therefore, we do not assume that we would observe the same response from student $kji$ if the treatment were assigned to some but not all of the students in school $kj$.

Our aim is to estimate the effect on test scores of attending a private voucher school as opposed to a public school for  comparable matched students $kji$ within matched schools $kj$. 
Therefore, this causal effect is defined for the population of students that are comparable in terms of their observed covariates and that are marginal in the sense that they may or not receive the treatment \citep{Rosenbaum:2012}. 
In other words, this effect is not about the students that will never be assigned to a private voucher school, or that will never be assigned to a public school, but instead about students that may be assigned to either of them as a function of their observed covariates.
As we explain in Section \ref{sec:overlap}, this estimand acknowledges the fact that not all the students are comparable in terms of observed covariates and it depends on the sample data \citep{crump2009}.
It also makes explicit that inferences to other target populations will require further modeling assumptions. 
In Section \ref{sec:overlap} we discuss how the estimand may change when the treated and control students are not comparable, and describe the population of students to which the results from the matched sample in principle generalize.

To identify this effect, we assume that the school level treatment assignment is strongly ignorable \citep{Rosenbaum:1983}.  Formally, the assumption implies that the cluster level treatment assignment is unconfounded,
\[
\Pr(Z_{kj} = 1|y_{Tkji}, y_{Ckji}, \mathbf{x}_{kji}, u_{kji}) = \Pr(Z_{kj} = 1| \mathbf{x}_{kji}),
\]
and probabilistic
\[
0 < \Pr(Z_{kj} = 1|y_{Tkji}, y_{Ckji}, \mathbf{x}_{kji}, u_{kji}) < 1,
\]
for each unit $kji$ in each cluster $kj$ (\citealt{imbens2015b}, Section 3.4).
Intuitively, this assumption implies that after matching for observed covariates there are no systematic, pretreatment differences in unobserved covariates between the treatment and control groups, and that every school has a non-zero probability of receiving treatment. In Section \ref{subsec_equivalence_sensitivity}, we assess the sensitivity of our findings to departures from the assumption that treatment assignment is unconfounded.

Our matched design is meant to replicate a paired group randomized controlled trial (RCT). In a group RCT, treatment assignment occurs at the cluster, here school, level, and all units within a cluster are assigned to receive treatment. In a paired group RCT, clusters are paired using baseline covariates prior to randomization.  However, the analogy is inexact, since we pair not only clusters but units within matched clusters. We do this to guarantee that paired clusters are balanced both in terms of covariates and sample size. Thus our design is comparable to a paired group RCT, where before randomization both clusters and units within clusters are paired. One reason to favor such a design is that randomization of clusters may not balance unit level covariates.  See \citet{Hansen:2008b} for an example where clustered allocation of treatment left imbalances in unit level covariates.  

\subsection{Longitudinal Census of Students and Schools}

In 1988, Chile introduced a national student assessment system known as the \emph{Sistema Nacional de Medici$\acute{o}$n de la Calidad de la Educaci$\acute{o}$n} or SIMCE. The SIMCE is an ``educational census.'' That is, in the SIMCE, the Chilean Ministry of Education collects data of all the students in fourth, eighth, tenth and eleventh grades to evaluate their performance in language, mathematics and sciences. The SIMCE data are collected roughly every two years from four different sources. First, data are collected from students, which includes test scores that are complemented with other student covariates such as gender. Second, both parents and teachers complete questionnaires. Finally for schools, student test scores are aggregated, and a few additional covariates are collected. Students are given unique identifiers which allows us to form a true panel over a two year period. Student records can also be linked to teacher, parent, and school level covariates. 

In our study, we use SIMCE data from years 2003, 2004 and 2006. From year 2003 we use school-level measurements of secondary-only schools before the hypothesized treatment intervention. From 2004, we use student-level data for students in primary-only public schools when they were in eighth grade. Finally, from 2006 we use student test scores on language and mathematics administered when students were in the tenth grade as our outcome measures. 

\subsection{Data Structure and Study Design}

In our study, the data structure and study design are intricately linked. We now outline how we constructed the match to fit the data structure. One advantage of our strategy is that we can tailor the statistical adjustment to fit the multilevel structure of the data, which is important since we have student, parent, teacher, and school level data. We perform two matches: one for students and one for schools. Next, we describe the covariates that form the student level match.

For each student with test scores observed in 2006, we match on student, parent, teacher, and primary school covariates from the SIMCE data collection in 2004. For the student match, we first list student level covariates. The key covariate here is student test scores from the 8th grade. In 8th grade students are tested on four topics: language, mathematics, social sciences, and natural sciences. The student level data also measures gender. For the student match, we also include three covariates from parents: income measured in six categories, father's education, and mother's education. We also link students to primary school level measures, and we match on primary school covariates in the student level match. At the primary school level, we match on a five category socio-economic status indicator for each school that is created by the Chilean Ministry of Education.  This five category indicator is constructed from questions based on parental education, family incomes in the school and an index of school vulnerability. We also use school level measures that are aggregates of data observed at other levels. As such, we match on average test scores for each primary school, the number of teachers and the number of enrolled students. Finally, in the teacher survey, teachers are asked what level of education they expect the majority of their students to achieve. Teachers responded using a five category scale that records responses from 8th grade to a college degree. We aggregate this measure and recorded the median for each primary school and use it in the student level match. 

The school match is based on secondary school data from 2003. Since the SIMCE forms a panel, we can match on characteristics of the secondary schools before any student is exposed to the treatment. That is, we match on the schools the students will attend using data from before they attend that school. For the school match, we match on enrollment, school level math and language test score averages, the percentage of female students in the school, average student income, urban versus rural status, and the same five category socio-economic status indicator for each school that is recorded for primary schools.\footnote{For student in secondary schools, the SIMCE only collects test scores on language and math.  At the primary school level, we have test scores for language, mathematics, social sciences, and natural sciences.} Some of these covariates are aggregates that we created from either student, teacher, or parent level data in 2003. We did not match on several other covariates that are also observed in the 2003 data.  These measures include whether teachers are allowed class preparation time, the proportion of teachers with a post-graduate diploma, the average teacher experience, the number of hours teachers worked per week.  We do not match on these covariates since they are plausibly part of the school level treatment.  Matching on such covariates would remove their effect on students from the final outcomes and thus could potentially attenuate the treatment effect.


We describe the matching algorithm in greater detail in Section~\ref{sec:algo}.
The match is based on integer programming which allows us to enforce different forms of balance for different covariates \citep{Zubizarreta:2012, zubizarreta_et_al_2014}. This is relevant since we tailored the constraints for each covariate. Here, we describe the different balance constraints we applied to each covariate. For the student level covariates, we applied a mean balance constraint to primary school test score measures, primary school enrollment, the number of teachers in the primary school, the average expected level of educational attainment, and the proportion of female student in the primary school. For student level test score measures, we enforced a constraint on the entire distribution via the Kolmogorov-Smirnov test statistic which is the maximum discrepancy in the empirical cumulative distribution functions. For the school level match, we enforced a mean balance constraint on secondary school test scores, missingness indicators for test scores, secondary school enrollment, income category, SES category, urban or rural status, and the proportion of female students in the secondary school.\footnote{\textcolor{black}{To be precise, we required the absolute differences in means in these covariates to differ at most by 0.1 standard deviations.  Please see Table \ref{meanbalsch} below.}}

For discrete student- and school-level covariates, we used an approximate fine balance constraint.  Under fine balance, we exactly balance covariates without exactly matching. Fine balance is achieved for discrete covariates by balancing the marginal distributions of covariates exactly in aggregate but without constraining who is matched to whom. We applied approximate fine balance to student sex, father and mother's education level, parental income categories, and primary school SES categories. See \cite{rosenbaum_et_al_2007} for a discussion of fine balance and \citet[Part II]{Rosenbaum:2010} for a discussion of different forms of covariate balance. In the following section we describe our approach to multilevel matching.

\singlespacing
\section{Dynamic and Integer Programming for Multilevel Matching}
\doublespace
\label{sec:algo}

The goal of our multilevel matching strategy is to find the largest sample of matched pairs of treated and control units within matched pairs of treated and control clusters that is balanced on the observed covariates. For assessing the sensitivity of results to the influence of unobserved covariates we use a sensitivity analysis proposed by \citet{rosenbaum_1987a, Rosenbaum:2002} and tailored to clustered treatment assignments by \citet{Hansen:2014} (see Section \ref{subsec_equivalence_sensitivity} of the paper). In our case study, units are students and clusters are schools, and, importantly, because results can be confounded both by student and school level covariates, we match pairs of students and schools to balance covariates at both levels. The basic tool that we use in our multilevel matching strategy is cardinality matching which we describe subsequently.

\subsection{Review of Cardinality Matching}

Common matching methods attempt to achieve covariate balance indirectly, by finding treated and control units that are close on a summary measure of the covariates such as the Mahalanobis distance or the propensity score (see \citealt{stuart_2010} and \citealt{lu_et_al_2011} for reviews). Unlike these matching methods, cardinality matching uses the original covariates to match units and directly balance their covariate distributions \citep{zubizarreta_et_al_2014}. Specifically, by solving an integer programming problem, cardinality matching finds the \emph{largest} matched sample that satisfies the investigator's specifications for covariate balance.
\textcolor{black}{Following \citet{Zubizarreta:2012}, these specifications for covariate balance may not only require balancing the means of the covariates, but perhaps also balancing entire distributions via fine balance \citep{rosenbaum_et_al_2007}, $x$-fine balance \citep{zubizarreta_et_al_2011}, and strength-$k$ \textcolor{black}{balance} \citep{hsu_et_al_2015}.}
For example, cardinality matching will find the largest sample of matched pairs in which all the covariates have differences in means smaller than one tenth of a standard deviation and the marginal distributions of nominal covariates of greater prognostic importance are perfectly balanced (fine balance).
In this manner, with cardinality matching subject matter knowledge about the research question at hand comes into the matching problem through the specifications for covariate balance, finding the largest matched sample that satisfies them.

\textcolor{black}{Cardinality matching is well suited for studying unit-level interventions but it is unclear how to use it for cluster-level interventions with any optimality guarantees (this same consideration applies to other matching methods such as optimal matching).} As we describe in Section \ref{subsecproc}, our multilevel matching strategy uses cardinality matching to match treated and control students across all the possible combinations of treated and control schools, and then uses a modified version of cardinality matching to match schools with the largest number of matched students. \textcolor{black}{As we discuss, our multilevel matching strategy can be used with other matching methods such as optimal matching to minimize the total sum of covariate distances between matched students within matched schools.}

\subsection{A Multistage Decision Strategy for Multilevel Matching}	\label{subsecproc}

Let $k_t \in \mathcal{K}_t = \{ 1, ..., K_t \}$ index the treated clusters and $k_c \in \mathcal{K}_c = \{ 1, ..., K_c \}$ denote the control clusters.
Let $i_{k_t}$ be treated unit $i$ in treated cluster $k_t$, with $i_{k_t} \in \mathcal{I}_{k_t} = \{ 1, ..., I_{k_t} \}$, and $i_{k_c}$ stand for control unit $i$ in control cluster $k_c$ with $i_{k_c} \in \mathcal{I}_{k_c} = \{ 1, ..., I_{k_c} \}$.
Put $\mathbf{x}_{k_t}$ for the vector of observed covariates of treated cluster $k_t$, and similarly write $\mathbf{x}_{i_{k_t}}$ for the observed covariates of treated unit $i_{k_t}$; analogous notation applies for control clusters and units.
Based on the unit-level covariates, calculate a distance $\delta_{i_{k_t}, i_{k_c}}$ between treated unit $i_{k_t}$ and control unit $i_{k_c}$ (for instance, this distance may be the robust Mahalanobis distance specified in Section 8.3 of \citealt{Rosenbaum:2010}).
Define $\mathcal{A}$ and $\mathcal{B}_{\boldsymbol{a}}$ as the sets of feasible solutions for the cluster- and unit-level matches within matched clusters (hence the subindex $\boldsymbol{a}$ in $\mathcal{B}_{\boldsymbol{a}}$).
In practice, $\mathcal{A}$ and $\mathcal{B}_{\boldsymbol{a}}$ are implemented as linear inequality constraints in an integer program and they enforce the investigator's requirements for covariate balance and matching structures at the cluster and unit levels respectively (for instance, $\mathcal{A}$ may require the means of the cluster covariates to be balanced and the matched groups to form pairs of clusters, and $\mathcal{B}_{\boldsymbol{a}}$ may require the marginal distributions of the unit covariates to be balanced and the matched groups to form pairs of units).
Importantly, since the requirements in $\mathcal{A}$ refer to clusters and those in $\mathcal{B}_{\boldsymbol{a}}$ refer to units, $\mathcal{A}$ and $\mathcal{B}_{\boldsymbol{a}}$ are disjoint.
Let $\mathcal{I}_{k_t}^{(m)}$ be the set of treated units matched in treated cluster $k_t$ and $\mathcal{I}_t^{(m)} = \bigcup_{k_t \in \mathcal{K}_t} \mathcal{I}_{k_t}^{(m)}$ be the set of treated units matched across all treated clusters.
Finally, let $\mathcal{K}_t^{(m)}$ be the set of matched treated clusters.

Building upon the framework of \citet{rosenbaum_2012a}, an optimal cardinality matching of units within clusters can be characterized by the quadruple $(\mathcal{K}_t^{(m)}, \alpha, \mathcal{I}_t^{(m)}, \beta)$ of assignments of clusters $\alpha: \mathcal{K}_t^{(m)} \to \mathcal{K}_c$ and units $\beta: \mathcal{I}_{k_t}^{(m)} \to \mathcal{I}_{k_c}$ that maximize the cardinality of the set of matched of units within matched clusters subject to the constraints in $\mathcal{A}$ and $\mathcal{B}_{\boldsymbol{a}}$, respectively.
If there are two cardinality matchings that satisfy the requirements in $\mathcal{A}$ and $\mathcal{B}_{\boldsymbol{a}}$, then we prefer one matching over the other if it has a larger cardinality, or, alternatively, if they both have the same cardinality, if it has a smaller sum of total distances between matched units.
Formally, we prefer the cardinality matching $(\mathcal{K}_t^{(m)}, \alpha, \mathcal{I}_t^{(m)}, \beta)$ to $(\tilde{\mathcal{K}}_t^{(m)}, \tilde{\alpha}, \tilde{\mathcal{I}}_t^{(m)}, \tilde{\beta})$, denoted by $(\mathcal{K}_t^{(m)}, \alpha, \mathcal{I}_t^{(m)}, \beta) \succ (\tilde{\mathcal{K}}_t^{(m)}, \tilde{\alpha}, \tilde{\mathcal{I}}_t^{(m)}, \tilde{\beta})$, if $| \mathcal{I}_t^{(m)} | > | \tilde{\mathcal{I}}_t^{(m)} |$, or alternatively if $|\mathcal{I}_t^{(m)}| = | \tilde{\mathcal{I}}_t^{(m)}|$ and $\sum_{i_{k_t} \in \mathcal{I}_t^{(m)}} \delta_{i_{k_t}, \beta(i_{k_t})} < \sum_{i_{k_t} \in \tilde{\mathcal{I}}_t^{(m)}} \delta_{i_{k_t}, \beta(i_{k_t})}$.
If $| \mathcal{I}_t^{(m)}| = |\tilde{\mathcal{I}}_t^{(m)}|$ and $\sum_{i_{k_t} \in \mathcal{I}_t^{(m)}} \delta_{i_{k_t}, \beta(i_{k_t})} = \sum_{i_{k_t} \in \tilde{\mathcal{I}}_t^{(m)}} $ $ \delta_{i_{k_t}, \beta(i_{k_t})}$, then we are indifferent between the two cardinality matchings and write $(\mathcal{K}_t^{(m)}, \alpha, \mathcal{I}_t^{(m)}, $ $\beta) $ $ \sim (\tilde{\mathcal{K}}_t^{(m)}, \tilde{\alpha}, \tilde{\mathcal{I}}_t^{(m)}, \tilde{\beta})$.
If we have either $(\mathcal{K}_t^{(m)}, \alpha, \mathcal{I}_t^{(m)}, \beta)$ $\succ (\tilde{\mathcal{K}}_t^{(m)}, \tilde{\alpha}, \tilde{\mathcal{I}}_t^{(m)}, \tilde{\beta})$ or $(\mathcal{K}_t^{(m)}, \alpha, $ $ \mathcal{I}_t^{(m)}, $ $\beta) $ $ \sim (\tilde{\mathcal{K}}_t^{(m)}, \tilde{\alpha}, $ $ \tilde{\mathcal{I}}_t^{(m)}, $ $\tilde{\beta})$, we write $(\mathcal{K}_t^{(m)}, \alpha, \mathcal{I}_t^{(m)}, \beta) \succsim (\tilde{\mathcal{K}}_t^{(m)}, \tilde{\alpha}, \tilde{\mathcal{I}}_t^{(m)}, \tilde{\beta})$.
Our optimal multilevel matching problem is the following.

\vspace{.4cm}
\begin{problem}\label{problem}
For given sets of cluster-level constraints $\mathcal{A}$ and unit-level constraints $\mathcal{B}_{\boldsymbol{a}}$, find a matching $(\mathcal{K}_t^{(m)}, \alpha, \mathcal{I}_t^{(m)}, \beta)$ that satisfies $\mathcal{A}$ and $\mathcal{B}_{\boldsymbol{a}}$ such that, for any other matching $(\tilde{\mathcal{K}}_t^{(m)}, \tilde{\alpha},$ $\tilde{\mathcal{I}}_t^{(m)}, \tilde{\beta})$ that also satisfies $\mathcal{A}$ and $\mathcal{B}_{\boldsymbol{a}}$, $(\mathcal{K}_t^{(m)}, \alpha, \mathcal{I}_t^{(m)}, \beta) \succsim (\tilde{\mathcal{K}}_t^{(m)}, \tilde{\alpha}, \tilde{\mathcal{I}}_t^{(m)}, \tilde{\beta})$.
\end{problem}

Intuition may suggest that the the best way to solve Problem \ref{problem} and match with multilevel data is first to match clusters and then within matched clusters to match units. In our case study, this would require first pairing schools and then, within pairs of schools, pairing students. However this strategy will not always find the largest matched sample that is balanced as two schools that are paired on their school level characteristics may have different student compositions so that when their students are paired it may result in a smaller sample size than optimal. For this reason, the optimal matching strategy needs to contemplate what is optimal both at the student and school levels simultaneously. Applying Bellman's \citeyearpar{bellman_1957} principle of optimality, the optimal matching strategy is first to match units across all the possible combinations of pairs of treated and control clusters, and, once all possible combinations of matched units are known, then to match clusters.

In abstract terms, the following algorithm and proposition state this.
To implement the optimal assignments $\alpha$ and $\beta$, let $a_{k_t, k_c} = 1$ if treated cluster $k_t$ is paired to control cluster $k_c$ and $a_{k_t, k_c} = 0$ otherwise; similarly let $b_{i_{k_t}, i_{k_c}} = 1$ if treated unit $i$ in treated cluster $k_t$ is paired to control unit $i$ in control cluster $k_c$, and $b_{i_{k_t}, i_{k_c}} = 0$ otherwise.

\vspace{.4cm}
\begin{algorithm}\label{algorithm}
For each of the possible $K_t \times K_c$ pairs of treated and control clusters, find the optimal cardinality matching of units that satisfies $\mathcal{B}_{\boldsymbol{a}}$.
This is, for each $k_t \in \mathcal{K}_t$ and each $k_c \in \mathcal{K}_c$ find $m_{k_t, k_c} = \max_{\boldsymbol{b}} \sum_{i_{k_t} \in \mathcal{I}_{k_t}} \sum_{i_{k_c} \in \mathcal{I}_{k_c}} b_{i_{k_t}, i_{k_c}}$ $\textrm{subject to}$ $\boldsymbol{b} \in \mathcal{B}_{\boldsymbol{a}}$.
Then find the optimal cardinality cluster matching that solves $\max_{\boldsymbol{a}} \sum_{k_t \in \mathcal{K}_t} \sum_{k_c \in \mathcal{K}_c} m_{k_t, k_c} a_{k_t, k_c}$ $\textrm{subject to}$ $\boldsymbol{a} \in \mathcal{A}$.
\end{algorithm}

\vspace{.4cm}
\begin{proposition}
Algorithm \ref{algorithm} solves the optimal multilevel cardinality matching problem \ref{problem}.
\end{proposition}
\begin{proof}
Let $f(\boldsymbol{a}, \boldsymbol{b})$ be the the total number of pairs of treated and control units matched by $\boldsymbol{b}$ within pairs of treated and clusters matched by $\boldsymbol{a}$.
In the abstract, in Problem \ref{problem} we want to maximize the function $f(\boldsymbol{a}, \boldsymbol{b})$ subject to the constraints $\mathcal{A}$ and $\mathcal{B}_{\boldsymbol{a}}$.
This is, find $\boldsymbol{a}$ and $\boldsymbol{b}$ to solve
\begin{equation} 
\max_{\boldsymbol{a}, \boldsymbol{b}} f(\boldsymbol{a}, \boldsymbol{b}) \textrm{ subject to } \boldsymbol{a} \in \mathcal{A}, \boldsymbol{b} \in \mathcal{B}_{\boldsymbol{a}}. \label{gen1}
\end{equation}
In a trivial way, we may solve (\ref{gen1}) by first solving
\begin{equation} 
g(\boldsymbol{a}) = \max_{\boldsymbol{b}} f(\boldsymbol{a}, \boldsymbol{b}) \textrm{ subject to } \boldsymbol{b} \in \mathcal{B}_{\boldsymbol{a}} \label{gen2}
\end{equation} 
for each $\boldsymbol{a} \in \mathcal{A}$, and then solving
\begin{equation} 
\max_{\boldsymbol{a}} g(\boldsymbol{a}) \textrm{ subject to } \boldsymbol{a} \in \mathcal{A}. \label{gen3}
\end{equation} 
While (\ref{gen2}) seems hard in general (because there are many possible choices of $\boldsymbol{b}$), the nested structure of the units-in-clusters problem makes it easier because $f(\boldsymbol{a}, \boldsymbol{b})$ separates into a sum of parts for cluster pairs because the constraint sets $\mathcal{A}$ and $\mathcal{B}_{\boldsymbol{a}}$ are disjoint.
Algorithm \ref{algorithm} does exactly this.
\end{proof}


In our case study, for each pairing of schools $\boldsymbol{a}$, we find the best pairing of students $\boldsymbol{b}$ within those schools (\ref{gen2}), and then pick the best pairing of schools with the associated best pairing of students for that pairing of schools (\ref{gen3}).
Again, while (\ref{gen2}) seems hard in general (because there are many possible student matches $\boldsymbol{b}$ across schools), the nested structure of the students-in-schools problem makes it easier because $f(\boldsymbol{a}, \boldsymbol{b})$ separates into a sum of parts for school pairs.
For example, if treated school $k_t$ is paired to control school $k_c$, then the contribution of schools $k_t$ and $k_c$ is the same of number of pairs regardless of how the other schools are paired.

With Algorithm \ref{algorithm}, the multilevel cardinality matching problem can be solved optimally by breaking it into simpler matching subproblems and recursively finding the optimal match.
This is an application of dynamic programming to matching in observational studies that takes advantage of the multilevel structure of the data (see \citealt{bertsekas_2005} for an extensive exposition of dynamic programming).

\subsection{Illustrative example}

We present a simple example to fix ideas.
In this example, there are 3 treated schools and 5 control schools indexed by $k_t \in \mathcal{K}_t = \{ 1, ..., 3 \}$ and $k_c \in \mathcal{K}_c = \{ 1, ..., 5\}$, respectively.
The first step of our multilevel matching strategy as implemented with cardinality matching is to find the student pair-matches across all the possible combinations of pairs of treated and control schools.
Here, since there are 3 treated schools and 5 control schools, we solve $3 \times 5 = 15$ student matching problems.\footnote{We match the students within treated school 1 and control schools $\{ 1, ..., 5\}$, treated school 2 and control schools $\{ 1, ..., 5\}$, and so on.}
For each of these 15 combinations of pairs of treated and control schools, we use cardinality matching and record the cardinality of the pair-matched students, $m_{k_t, k_c}$ (so $m_{1, 1}$ is the number of students that were pair-matched between treated school and control school 1, $m_{1, 2}$ is the number of pair-matched students between treated school 1 and control school 2, and so on). The size of the cardinalities $m_{k_t, k_c}$ will depend on the student covariate balance constraints, denoted $\mathcal{B}_{\boldsymbol{a}}$, and the cardinality may be zero if the balance constraints cannot be met.
Once we obtain all the $m_{k_t, k_c}$'s, the second step is to solve the school level matching problem, which maximizes the total number of pair-matched students within pair-matched schools subject to school level balance constraints, $\mathcal{A}$.
This problem is $\max_{\boldsymbol{a}} \sum_{k_t \in \mathcal{K}_t} \sum_{k_c \in \mathcal{K}_c} m_{k_t, k_c} a_{k_t, k_c}$ $\textrm{subject to}$ $\boldsymbol{a} \in \mathcal{A}$. 
Solving the problem backwards ---from the students to the schools, instead of the opposite--- ensures that we find the overall optimum number of student matches. 
If we were to do the opposite (this is, first match schools and then within matched schools match students), then we would likely find a suboptimal solution as we show in the comparison study in Section \ref{sec:comparison} below.

 

\subsection{Additional considerations}

Note that if we had three or more levels (such as students within schools within districts), then our multilevel matching strategy would extend naturally. With $l$ levels, the strategy would require first matching the lowest level $l$ under the assumptions that levels $l-1, l-2, ..., 1$ have been matched optimally, to then (once the matches at level $l$ are completed) matching level $l-1$ under the assumptions that levels $l-2, l-3..., 1$ have been matched optimally, and so on.

The multilevel matching problem above is formulated to maximize the size of the matched sample in terms of students, but it can be easily modified to maximize a weighted combination of students and clusters; namely, $\max_{\boldsymbol{a}} \sum_{k_t \in \mathcal{K}_t} \sum_{k_c \in \mathcal{K}_c} m_{k_t, k_c} a_{k_t, k_c} + \lambda \sum_{k_t \in \mathcal{K}_t} \sum_{k_c \in \mathcal{K}_c} a_{k_t, k_c}$ $\textrm{subject to}$ $\boldsymbol{a} \in \mathcal{A}$, for a suitable scalar $\lambda$.
Also, the multilevel matching problem can be formulated to minimize a covariate distance between students by using optimal matching to solve each of the component problems.
In all these cases, the optimal matching strategy is to match backwards, by first matching the units and then matching the clusters.

We have discussed the optimality of our multilevel matching strategy from a mathematical programming standpoint. From a statistical standpoint, this matching approach as implemented with cardinality matching gives priority to reducing bias over increasing precision, and is optimal in the sense that given a matching structure (e.g., pair matching) and degree of bias reduction (these are the covariate balancing requirements) it maximizes the accuracy of the study (by maximizing the size of the matched sample). More precisely, extending the argument in \citet{kilcioglu2016}, under a constant additive treatment effect model this approach minimizes the variance of a difference-in-means effect estimator if the intracluster correlation is less than one. In general, the statistical optimality of matching methods is an important area that remains to be studied within a formal statistical framework.

\section{Evaluation of the Matched Sample}	
\label{sec:assessment}
\vspace{-1.5cm} 

\textcolor{black}{
\subsection{Covariate Balance and Sample Size}	
}

After applying basic exclusion criteria, there are 64245 students in 517 schools, 150 private subsidized and 367 public schools (henceforth treated and control schools respectively). Out of the 64245 students, 15682 students are from treated schools and 48563 are from control schools. Using \textcolor{black}{the} multilevel matching strategy \textcolor{black}{above}, we matched in two stages within similar groups regions of the country (namely, regions I-III, IV-V, VI-VII, VIII, IX, X-XII and the Metropolitan region).

At the student level, we used cardinality matching to find the largest balanced sample of pairs of students across all the possible combinations of pairs of schools within the groups of regions.
\textcolor{black}{In each of these matches we balanced the means of 19 covariates (including student test scores, school test scores, and indicators for socioeconomic status and expected educational achievement; see Table \ref{meanbalstu} for details),\footnote{\textcolor{black}{To be precise, we required the absolute differences in means in these covariates to differ at most by 0.1 standard deviations.  The matched sample that we found did not only satisfy these mean balance constraints but achieved absolute differences in means smaller than 0.05 standard deviations. Please see Table  \ref{meanbalstu} below.}} approximately fine balanced the marginal distribution of four nominal covariates (sex, mother and father education, and household income; see Table \ref{approxfinebalstu}) and balanced marginal the distribution of two continuous covariates (the sum of the test scores in language and mathematics at baseline).}
Figure \ref{joint} shows not only that the marginal distributions of the baseline test scores are very closely balanced after matching but also their joint distribution.
As a matter of fact, the 95\% bivariate normal density contours are almost indistinguishable after matching.

\begin{table}[htbp]
\caption{Covariate balance at the student level after matching.  All the covariates are measured in 2004. \textcolor{black}{The last column shows standardized differences in means.}}
\textcolor{black}{
\label{meanbalstu}
\vspace{.2cm}
\begin{center}
\begin{tabular}{@{\extracolsep{2pt}}lrrr}
\hline
\multicolumn{1}{c}{Covariate} & \multicolumn{2}{c}{Mean} & \multicolumn{1}{c}{Std. dif.} \\  
\cline{2-3}
& \multicolumn{1}{c}{Private} & \multicolumn{1}{c}{Public} & \multicolumn{1}{c}{} \\ 
\hline
Language score & 244.84 & 244.75 & 0.00 \\ 
Mathematics score & 245.71 & 244.92 & 0.02 \\ 
Natural science score & 248.74 & 248.57 & 0.00 \\ 
Social science score & 244.38 & 244.49 & -0.00 \\ 
School language score & 238.84 & 238.42 & 0.02 \\ 
School mathematics score & 239.89 & 239.32 & 0.03 \\ 
School female proportion & 0.51 & 0.51 & -0.01 \\ 
School number of students & 84.23 & 84.17 & 0.00 \\ 
School teacher to student ratio & 8.22 & 8.14 & 0.04 \\ 
Urban area & 0.85 & 0.85 & -0.00 \\ 
Socioeconomic status A & 0.13 & 0.13 & -0.01 \\ 
Socioeconomic status B & 0.57 & 0.56 & 0.01 \\ 
Socioeconomic status C & 0.28 & 0.29 & -0.01 \\ 
Socioeconomic status D & 0.02 & 0.01 & 0.02 \\ 
Expected education: primary & 0.01 & 0.01 & -0.00 \\ 
Expected education: secondary, technical-professional & 0.72 & 0.69 & 0.05 \\ 
Expected education: secondary, scientific-humanities & 0.16 & 0.17 & -0.04 \\ 
Expected education: technical-professional & 0.10 & 0.10 & -0.01 \\ 
Expected education: college & 0.01 & 0.02 & -0.05 \\ 
   \hline
\end{tabular}
\end{center}
}
\end{table}

\begin{table}[htbp]
\caption{Balance for nominal covariates at the student level.  All the covariates are measured in 2004.  \textcolor{black}{Approximate fine} balance for sex,
type of education of the mother and father, and household income category.  The tabulated values are 
counts of the number of students in each category.  In addition, matching was exact for groups of counties (not shown here).}
\textcolor{black}{
\label{approxfinebalstu}
\vspace{.2cm}
\begin{center}
\begin{tabular}{@{\extracolsep{2pt}}lrr}
\hline
\multicolumn{1}{c}{Covariate} & \multicolumn{2}{c}{School type} \\
\cline{2-3}
 & Private & Public \\ 
\hline
Sex & & \\
\hspace{0.35cm} Male & 4634 & 4638 \\ 
\hspace{0.35cm} Female & 4168 & 4164 \\ 
\hline
Mother education & & \\
\hspace{0.35cm} Primary school & 4016  & 4026 \\ 
\hspace{0.35cm} Secondary school & 2578  & 2577 \\ 
\hspace{0.35cm} Technical &  176   &  173 \\ 
\hspace{0.35cm} College or higher &  73    &  75 \\ 
\hspace{0.35cm} Missing & 1959  & 1951 \\
\hline
Father education & & \\ 
\hspace{0.35cm} Primary school & 3538 & 3552  \\ 
\hspace{0.35cm} Secondary school & 2732 & 2724 \\ 
\hspace{0.35cm} Technical &  196   &  195 \\ 
\hspace{0.35cm} College or higher &  130   &  131 \\ 
\hspace{0.35cm} Missing & 2206  & 2200 \\
\hline
Household income category (in 1000 pesos) & & \\
\hspace{0.35cm} $[0, 100)$ & 251   & 248 \\ 
\hspace{0.35cm} $[100, 200]$ & 3356 & 3373 \\ 
\hspace{0.35cm} $(200, 400]$ & 3346  & 3346 \\ 
\hspace{0.35cm} $(400, 600]$ & 974   & 971 \\ 
\hspace{0.35cm} $(600, 1400]$ &  289  &  291 \\ 
\hspace{0.35cm} $>1400$ &   206   &  203 \\ 
\hspace{0.35cm} Missing & 380   & 370 \\  
\hline
\end{tabular}
\end{center}
}
\end{table}%

\begin{figure}[h!]
\caption{Distribution of student test scores at baseline after matching.  The baseline test scores are measured in 2004. The ellipses trace the 95\% bivariate normal density contours of the joint distributions of test scores for the matched treated and control units.  The contours are almost identical showing that not only marginal distributions of the test scores are very closely balanced but also their joint distribution.}
\begin{center}
\includegraphics[scale=0.9]{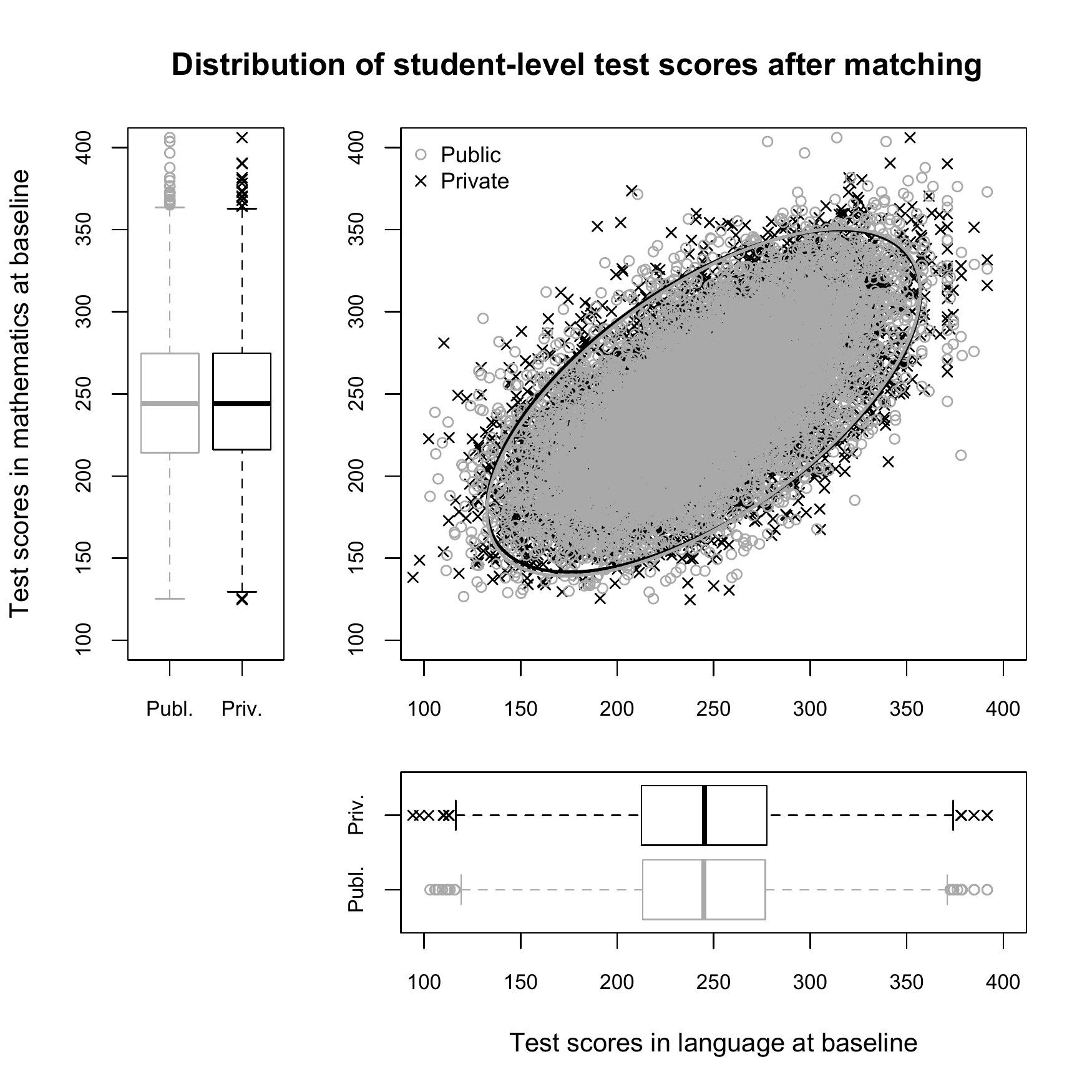}
\end{center}
\label{joint}
\end{figure}

At the school level, we used the modification of cardinality matching in the second stage of Algorithm \ref{algorithm} and mean balanced 16 other covariates: percentage female, total enrollment, language and math scores (plus indicators for missing values), urban area, parental income categories (1-5), and socioeconomic groups (A-D). Again, covariates were exact matched for the 7 region groups \textcolor{black}{(for this, we basically solved different matching problems for each of these region groups)}. 
We balanced all covariates with and without weighting for the size of the school; see Table \ref{meanbalsch}.  Note that after matching all the differences in means are smaller than 0.05 standard deviations.
In this way, we matched \textcolor{black}{8802} students in \textcolor{black}{4401} pairs, and \textcolor{black}{280} schools in \textcolor{black}{140} pairs. 
In this match, \textcolor{black}{all the available regions} of the country are represented in both the treatment and control groups.

\begin{table}[htbp]
\caption{Covariate balance at the school level after matching.  
Both means and standardized differences are weighted by the number of students in each school.
\textcolor{black}{The last column shows standardized differences in means.}
}
\label{meanbalsch}
\vspace{.2cm}
\centering
\textcolor{black}{
\begin{tabular}{@{\extracolsep{2pt}}lrrr}
\hline
\multicolumn{1}{c}{Covariate} & \multicolumn{2}{c}{Mean} & \multicolumn{1}{c}{Std. dif.} \\  
\cline{2-3}
& \multicolumn{1}{c}{Private} & \multicolumn{1}{c}{Public} & \multicolumn{1}{c}{} \\ 
\hline  
  Female proportion & 0.46 & 0.49 & -0.09 \\ 
  Number of students & 193.38 & 193.38 & 0.00 \\ 
  Language score & 241.92 & 239.99 & 0.10 \\ 
  Mathematics score & 231.21 & 228.92 & 0.09 \\ 
  Language score missing & 0.02 & 0.02 & -0.05 \\ 
  Mathematics score missing & 0.02 & 0.02 & -0.05 \\ 
  Urban area & 0.92 & 0.94 & -0.09 \\ 
  Income category 1 & 0.20 & 0.24 & -0.08 \\ 
  Income category 2 & 0.67 & 0.66 & 0.02 \\ 
  Income category 3 & 0.12 & 0.10 & 0.05 \\ 
  Income category 4 & 0.01 & 0.00 & 0.08 \\ 
  Income category 5 & 0.00 & 0.00 & 0.00 \\ 
  Socioeconomic status A & 0.30 & 0.35 & -0.10 \\ 
  Socioeconomic status B & 0.55 & 0.54 & 0.03 \\ 
  Socioeconomic status C & 0.15 & 0.12 & 0.09 \\ 
  Socioeconomic status D & 0.00 & 0.00 & 0.00 \\
\hline
\end{tabular}
}
\end{table}

\subsection{Comparison to Myopic Matching}	
\label{sec:comparison}

We now compare the performance of our dynamic matching strategy to \textcolor{black}{an alternative} strategy which first matches schools followed by matching students within these matched schools. We implement this alternative matching strategy using two \textcolor{black}{different} methods: optimal matching (using a robust Mahalanobis distance matrix with calipers for violations on the propensity score; see \citealt{rosenbaum1985} and Section 8.3 of \citealt{Rosenbaum:2010}) and cardinality matching, using the approximation scheme in \citet{zubizarreta2016}.\footnote{This approximation scheme solves a linear program relaxation of the integer programming problem in cardinality matching and then rounds the solution by solving a linear program again or using a more specialized network algorithm.  This approximation may violate to some extent some of the balancing constraints but it runs quickly (in polynomial time), and in many applications the violations to the balancing constraints are not substantial.} We call \textcolor{black}{these two more standard strategies} ``myopic'' because they make locally optimal matching decisions at the school and student stages separately rather than collectively in view of the global optimum. The results from this comparison are summarized in Table \ref{tab1} below.


By construction, myopic optimal matching uses all the available observations, but results in substantial imbalances at both the school and student level. In fact, out of the 17 school-level covariates, 4 of them have \textcolor{black}{differences in means greater than 0.1 standard deviations} after matching, and the total variation distance \textcolor{black}{(defined as the sum of the treated-minus-control absolute value differences in percentage points across all categories of nominal covariates; \citealt{pimentel2015}) for the marginal distributions of the covariates in Table \ref{approxfinebalstu} is more than 10 times larger with optimal matching than with the cardinality matching methods.}
Myopic cardinality matching achieves good covariate balance both at the school and student levels, but matches considerably fewer students: 4913 students instead of the 8802 students matched by our strategy, which illustrates why working backwards from the optimal solutions is better than working forward in a greedy sense. \textcolor{black}{This demonstrates that our dynamic strategy to multilevel matching reduces imbalances and uses a larger portion of the sample.}

\begin{table}
\begin{center}
\caption{Comparison of matching methods.}
\label{tab1}
\begin{tabular}{@{\extracolsep{2pt}}lcccccc@{}}
\hline
\multicolumn{1}{c}{} & \multicolumn{2}{c}{Matches} & \multicolumn{2}{c}{Imbalances} & \multicolumn{2}{c}{Speed}  \\
\cline{2-3}
\cline{4-5}
\cline{6-7}
\multicolumn{1}{c}{\multirow{2}{*}{Method}} & \multirow{2}{*}{Schools} & \multirow{2}{*}{Students} & Means & TV & Problems & Time  \\
 &  &  & schools & students & solved & (min) \\
\hline
Myopic optimal & \multirow{2}{*}{149} & \multirow{2}{*}{9643} & \multirow{2}{*}{4} & \multirow{2}{*}{0.179} & \multirow{2}{*}{149} & \multirow{2}{*}{0.48}  \\ 
matching &  &  &  &  &  &   \\
Myopic cardinality  & \multirow{2}{*}{143} & \multirow{2}{*}{4913} & \multirow{2}{*}{0} & \multirow{2}{*}{0.016} & \multirow{2}{*}{143} & \multirow{2}{*}{1.37}  \\
matching$^\dagger$ &  &  &  &  &  &    \\
Dynamic cardinality & \multirow{2}{*}{140} & \multirow{2}{*}{8802} & \multirow{2}{*}{0} & \multirow{2}{*}{0.011} & \multirow{2}{*}{10261} & \multirow{2}{*}{148.43}  \\  
matching$^\dagger$ &  &  &  &  &  &    \\
\hline
\end{tabular}
\end{center}
\vspace{1mm}
\footnotesize{\emph{Note 1:} For optimal matching, we used a robust Mahalanobis distance matrix \citep{Rosenbaum:2010} with calipers for violations on the propensity score (\citealt{rosenbaum1985}).
For cardinality matching (\scriptsize$\dagger$\footnotesize) we used the approximation scheme in \citet{zubizarreta2016}.
In both cases, we used the \texttt{R} package \texttt{designmatch}.	\\
\emph{Note 2:} 
Under ``Imbalances,'' the first column (``Means schools'') shows the number of covariates that have imbalance in means at the school level (there are 17 school covariates in total).
In the second column, ``TV'' denotes the total variation distance for the marginal distributions of key covariates at the student level \citep{pimentel2015}.	\\
\emph{Note 3:}
Under ``Speed,'' the first column shows the total number of student matching problems solved, and the second column shows the total time in minutes they took without parallelizing.  On average, with the three matching methods each student matching problem took less than a second to be solved.
}
\end{table}
\normalsize

In terms of speed, the dynamic matching method takes longer. \textcolor{black}{On a standard desktop computer \textcolor{black}{(with a 3.4 GHz Intel Core i5 processor, 16 GB 1600 MHz DDR3 of memory, and the OS X 10.10.4 operating system)}, the dynamic method takes nearly two and a half hours, as opposed to the 0.48 and 1.17 minutes needed for the two myopic methods.} \textcolor{black}{The dynamic method requires more computing time,} since it is considering all the possible combinations of pairs of treated and control schools within groups of regions; in other words, it is solving 10261 student matching problems as opposed to 149 and 143 with the other methods.  On average, with the dynamic method, as well as with the other matching methods, each student matching problem takes less than one second to be solved. Again, the difference in time is driven by the number of student matching problems that each method is considering. \textcolor{black}{To further reduce the computing time required for the dynamic matching}, \textcolor{black}{one could find the student level matches} in parallel by separating all the possible pairs of treated and control schools into smaller mutually exclusive but exhaustive pairs of treated and control schools.

%

\vspace{-1cm}
\textcolor{black}{
\subsection{Limited Overlap and the Estimand}	
}
\label{sec:overlap}

In observational studies, \textcolor{black}{a common limitation encountered in practice is limited overlap or lack of common support of the covariate distributions across the treatment samples.} \textcolor{black}{As \citet{crump2009} note,} lack of common support can lead to estimates that are highly biased, too variable, and overly sensitive to model misspecification. \textcolor{black}{When faced with limited overlap,} investigators often ``trim'' the samples and restrict their analyses to subsamples that have common support.  For example, \citet{dehejia1999} discard the control units with an estimated propensity score smaller (greater) than the minimum (maximum) one of the treated units. Formal methods for trimming are proposed by \citet{crump2009} and \citet{Rosenbaum:2012}. However, most of these methods address the problem of limited overlap using the propensity score or another summary of the covariates such as the Mahalanobis distance. 
\textcolor{black}{In contrast, with our multilevel matching strategy, cardinality matching} addresses the common support problem without resorting to a summary of the covariates, as it directly finds the largest subsamples of treated and control units that meant the common support or balance constraints set by the investigator.	
 
With all these methods, restricting the analysis to the samples of treated and control units that overlap (or, ultimately, that are balanced) changes the estimand \textcolor{black}{such that it only applies to the population of \textcolor{black}{comparable or marginal units which} may or may not receive the treatment \textcolor{black}{\citep{Rosenbaum:2012}}.}  More specifically, the estimands that result from trimming units will be a more local versions of commonly used estimands such as the average treatment effect (ATE) or average treatment \textcolor{black}{effect} on the treated (ATT), and will depend on the sample data. \textcolor{black}{Changing the estimand in this way acknowledges} the inherent limitations imposed by the data. It also makes explicit that inferences to other target populations will require further modeling assumptions \citep{crump2009}. In other words, this approach places greater importance on internal as opposed to external validity \citep{shadish2002}.



In our study, we estimate the effect of a cluster-level treatment on a sample of comparable units within comparable clusters. 
Specifically, we estimate the effect of attending private subsidized (voucher) schools instead of public schools on a sample of students by selecting both comparable students and schools.  If we only selected schools to ensure common support and covariate balance, our estimand would be the effect of a school-level intervention on the students of a sample of comparable schools. However, since we also selected students (to ensure comparability of the students within the schools), our estimand is the effect of a school-level intervention on a sample of comparable students within comparable schools.


To obtain a basic understanding of how the matched sample differs from the larger population, we describe the samples of matched treated and control students and compared them to the full samples of students \citep{silber2015}. Table \ref{tab:desc} compares the means of the covariates in the samples of students before matching (``All'') and in the unmatched and matched samples. For the students in private schools, there are a number of significant differences between the matched and unmatched samples (not denoted in the table for clarity in the exposition). The largest differences are in urban area, with a lower proportion of the unmatched students living in urban areas, and in socieconomic status, with the unmatched students having higher a socieconomic status on average. These differences are less marked when comparing the matched and complete (``All'') samples of students in private schools. In principle, our results could be generalized to a population of students with characteristics similar to the ones of the matched samples, which do not differ much from the original sample of students that that under the private subsidized school treatment. However, such generalizations from the sample to a larger population would require additional assumptions and methods \citep{stuart2011use,hartman2015sate}.

\begin{sidewaystable}[ph!]
\caption{Description of the samples of students before and after matching.}
\label{tab:desc}
\vspace{2mm}
\centering
\begin{tabular}{@{\extracolsep{2pt}}lrrrrrr}
  \hline
 \multicolumn{1}{c}{Covariate} & \multicolumn{3}{c}{Private} & \multicolumn{3}{c}{Public} \\ 
 \cline{2-4}
 \cline{5-7} 
 & All & Unmatched & Matched & Matched & Unmatched & All \\ 
  \hline
  Language score & 244.26 & 243.51 & 244.84 & 244.75 & 247.67 & 247.14 \\ 
  Mathematics score & 245.04 & 244.19 & 245.71 & 244.92 & 247.28 & 246.85 \\ 
  Natural science score & 248.19 & 247.47 & 248.74 & 248.57 & 250.73 & 250.34 \\ 
  Social science score & 244.06 & 243.64 & 244.38 & 244.49 & 246.69 & 246.29 \\ 
  School language score & 239.39 & 240.09 & 238.84 & 238.42 & 240.59 & 240.19 \\ 
  School mathematics score & 240.46 & 241.20 & 239.89 & 239.32 & 241.94 & 241.47 \\ 
  School female proportion & 0.52 & 0.52 & 0.51 & 0.51 & 0.53 & 0.53 \\ 
  School number of students & 83.49 & 82.54 & 84.23 & 84.17 & 82.01 & 82.40 \\ 
  School teacher to student ratio & 8.18 & 8.14 & 8.22 & 8.14 & 7.96 & 7.99 \\ 
  Urban area & 0.82 & 0.79 & 0.85 & 0.85 & 0.81 & 0.82 \\ 
  Socioeconomic status A & 0.15 & 0.17 & 0.13 & 0.13 & 0.16 & 0.16 \\ 
  Socioeconomic status B & 0.55 & 0.52 & 0.57 & 0.56 & 0.51 & 0.52 \\ 
  Socioeconomic status C & 0.28 & 0.27 & 0.28 & 0.29 & 0.31 & 0.31 \\ 
  Socioeconomic status D & 0.03 & 0.04 & 0.02 & 0.01 & 0.02 & 0.02 \\ 
  Expected education: primary & 0.01 & 0.02 & 0.01 & 0.01 & 0.01 & 0.01 \\ 
  Expected education: secondary, technical-professional & 0.69 & 0.66 & 0.72 & 0.69 & 0.61 & 0.62 \\ 
  Expected education: secondary, scientific-humanities & 0.17 & 0.19 & 0.16 & 0.17 & 0.23 & 0.22 \\ 
  Expected education: technical-professional & 0.11 & 0.11 & 0.10 & 0.10 & 0.13 & 0.12 \\ 
  Expected education: college & 0.02 & 0.02 & 0.01 & 0.02 & 0.02 & 0.02 \\ 
   \hline
\end{tabular}
\end{sidewaystable}

\section{Outcome Analyses}
\label{sec:out}
\textcolor{black}{Having found the matched sample, we now estimate the effect of attending a private voucher school, test its significance, and assess the robustness of the findings to biases due to unobserved covariates.}
\textcolor{black}{We first assume that treatment is as-if randomly assigned to clusters conditional on the matched pairs \citep{Small:2008b} and explore whether our results would differ if we relax this assumption using a sensitivity analysis for clustered observational studies developed by \citet{Hansen:2014}.}

\singlespacing
\subsection{Randomization Inference When Treatment is Assigned at the School Level}
\doublespace

Following \citet{Small:2008b}, we collect in the set $\mathbf{\Omega}$ the $2^K$ treatment assignments for all $2K$ clusters, $\mathbf{Z} = (Z_{11}, Z_{12}, \dots, Z_{K2})^T$. If the probability of receiving treatment is equal for each school in each matched pair, then the conditional distribution of $\mathbf{Z}$ given that there is exactly one treated school in each pair equals the randomization distribution with $\Pr(Z_{kj} = 1 | y_{Tkji}, y_{Ckji}, \mathbf{x}_{kji}, u_{kji}, \mathbf{\Omega}) = 1/2$ for each school $j$ in pair $k$. Write $\mathbf{Y} = (Y_{111}, \dots, Y_{K2,n_{k2}})^T$ for the $N = \sum_{k,j}n_{k,j}$ dimensional vector of observed responses with the same notation for $\mathbf{y}_c$, which are potential responses under control.

We denote $T = t(\mathbf{Z}, \mathbf{Y})$ as our test statistic. Under the sharp null hypothesis of no treatment effect, $\mathbf{Y} = \mathbf{y}_c$, and therefore $T = t(\mathbf{Z}, \mathbf{y}_c)$. If treatment were randomly assigned within matched pairs, then
$
\Pr\{ t(\mathbf{Z}, \mathbf{Y}) \geq v | y_{Tkji}, y_{Ckji}, \mathbf{x}_{kji}, u_{kji}, \mathbf{\Omega}\} = \Pr\{ t(\mathbf{Z}, \mathbf{y}_c) \geq v | y_{Tkji}, y_{Ckji}, \mathbf{x}_{kji}, u_{kji}, \mathbf{\Omega}\}
$
with $\Pr(\mathbf{Z} = \mathbf{z}|y_{Tkji}, y_{Ckji}, \mathbf{x}_{kji}, u_{kji}, \mathbf{\Omega}) = 1/|\mathbf{\Omega}|$ for some cutoff value $v$.

For $T$, we use a test statistic from \citet{Hansen:2014}. This test statistic is a function of $q_{kji}$, which is a score or rank given to $Y_{kji}$, so that under the null hypothesis, the $q_{kji}$ are functions of the $y_{Ckji}$ and $\mathbf{x}_{kji}$, and they do not vary with $Z_{kj}$. To make $q_{kji}$ resistant to outliers, we use the ranks of the residuals when $Y_{kji}$ is regressed on the student level covariates using Huber's method of m-estimation following \citet{Small:2008b}. The test statistic $T$ is 
\[
T = \sum_{k=1}^K B_k Q_k
\]
where
\[
B_k = 2Z_{k1} - 1 = \pm1, \;\;\; Q_k = \frac{w_k}{n_{k1}}\sum_{i=1}^{n_{k1}}q_{k1i} - \frac{w_k}{n_{k2}}\sum_{i=1}^{n_{k2}}q_{k2i}.
\]

\citet{Hansen:2014} show that $T$ is the sum of $k$ independent random variables each taking the value $\pm Q_k$ with probability $1/2$, so $E(T) = 0$ and $\mathrm{var}(T) = \sum_{k=1}^K Q_k^2$.  The central limit theorem implies that as $K \rightarrow \infty$, then $T/\sqrt{\mathrm{var}(T)}$ converges in distribution to the standard Normal distribution. In the above equation, $w_k$ defines the weights which are a function of $n_{kj}$. \citet{Hansen:2014} discuss possible choices for $w_k$.  One possibility is to use constant weights, $w_k \propto 1$.\footnote{Another possibility is to use weights that are proportional to the total number of students in a matched cluster pair: $w_k \propto n_{k1} + n_{k2}$ or $w_k = (n_{k1} + n_{k2})/\sum_{\ell=1}^K(n_{\ell1} + n_{\ell2})$. An analysis, with these weights did not alter the results.}

If we test the hypothesis of a shift effect instead of the hypothesis of no effect, we can apply the method of \citet{Hodges:1963} to estimate the private school effect. 
The Hodges-Lehmann (HL) estimate of $\tau$ is the value of $\tau_0$ that when subtracted from $Y_{kji}$ makes $T$ as as close as possible to its \textcolor{black}{expectation under the null}.  
Intuitively, the point estimate $\hat{\tau}$ is the value of $\tau_0$ such that $T$ equals $0$ when $T_{\tau_0}$ is computed from $Y_{kji} - Z_{kj}\tau_0$. Using constant effects is convenient, but this assumption can be relaxed; see \citet{Rosenbaum:2003}. If the treatment has an additive effect, $Y_{kji} = y_{Ckji} + \tau$ then a 95\% confidence interval for the additive treatment effect is formed by testing a series of hypotheses $H_0: \tau = \tau_0$ and retaining the set of values of $\tau_0$ not rejected at the 5\% level. 

\subsection{Comparative Effectiveness of Public Versus Private Subsidized Schools}

We now test the hypothesis of no effect for private subsidized schools on test scores. We measure test scores using an additive measure of language and mathematics scores. The approximate one-sided $p$-value for the test of the sharp null hypothesis is 0.207. If there are no hidden confounders, the point estimate of the private subsidized treatment effect is $\hat{\tau} =$ 2.55 with a 95\% confidence interval of -3.78 and 8.72. Thus, we cannot reject the hypothesis that attending a private subsidized school has no effect on test scores.  In the education literature, an effect size of 0.20 of a standard deviation is considered to be an educationally meaningful effect. Our estimated treatment effect is only 0.027 of a standard deviation. Therefore, our estimated treatment effect is well below the threshold for a meaningful effect.  We next explore the likelihood that bias from a hidden confounder masks a treatment effect.

 
\subsection{Test of Equivalence and Sensitivity Analysis}	
\label{subsec_equivalence_sensitivity}

In an observational study, one concern is that bias from a hidden covariate can give the impression that a treatment effect exists when in fact no effect is present. Bias from hidden confounders can also mask an actual treatment effect leaving the investigator to conclude there is no effect when in fact such an effect exists. We explore this possibility using a test of equivalence and a sensitivity analysis \citep{Rosenbaum:2008a,Rosenbaum:2009b,Rosenbaum:2010}. 

Above we were unable to reject the null hypothesis that $\tau = 0$ for all students. Next, we apply a test of equivalence to test the hypotheses that $\tau$ is not small. Under a test of equivalence, we test the following null hypothesis $H_{\neq}^{(\delta)}:|\tau| > \delta$. Rejecting $H_{\neq}^{(\delta)}$ provides a basis for asserting with confidence that $|\tau| < \delta$. $H_{\neq}^{(\delta)}$ is the union of two exclusive hypotheses: $\overleftarrow{H}_{0}^{(\delta)}:\tau \leq -\delta$ and $\overrightarrow{H}_{0}^{(\delta)}: \tau \geq \delta$, and $H_{\neq}^{(\delta)}$ is rejected if both $\overleftarrow{H}_{0}^{(\delta)}$ and $\overrightarrow{H}_{0}^{(\delta)}$ are rejected \citep{Rosenbaum:2009b}. We can apply the two tests without correction for multiple testing since we test two mutually exclusive hypotheses. 

With a test of equivalence, it is not possible to demonstrate a total absence of effect, but if this were a randomized trial, we could test that our estimated effect is not as large as $\delta$ by rejecting $H_{\neq}^{(\delta)}:|\tau| > \delta$. Since the treatment was not randomly assigned, it may be the case that we reject the null hypothesis of equivalence due to hidden confounding.  However, using a sensitivity analysis we may find evidence that the test of equivalence is insensitive to biases from nonrandom treatment assignment. 

 Thus far we have assumed that within matched pairs, receipt of the treatment is effectively random conditional on the matches. We consider how sensitive our conclusions are to violations of this assumption using a model of sensitivity analysis discussed in \citeauthor{Rosenbaum:2002} (2002, chapter 4). In our study, matching on observed covariates $\mathbf{x}_{kji}$ made students more similar in their chances of being exposed to the treatment. However, we may have failed to match on an important unobserved covariate $u_{kji}$ such that $\mathbf{x}_{kji} = \mathbf{x}_{kji'} \; \forall \; k, j, i, i'$, but possibly $u_{kji} \neq u_{kji'}$.  If true, the probability of being exposed to treatment may not be constant within matched school pairs. To explore this possibility, we use a sensitivity analysis. First, define $\pi_{k}$ as the probability that student $i$ in pair $k$ was treated. For two matched students in pair $k$, say $i$ and $i^\prime$, because they have the same observed covariates $\mathbf{x}_{kji} = \mathbf{x}_{kji^\prime}$ it may be true that $\pi_{k} = \pi_{k^\prime}$. However, if these two students differ in an unobserved covariate, $u_{kji} \neq u_{kji^\prime}$, then these two students may differ in their odds of being exposed to the private school treatment by at most a factor of $\Gamma \geq 1$ such that 

\begin{equation}
\frac{1}{\Gamma} \leq \frac{\pi_{k}/(1 - \pi_{k'})}{\pi_{k'}/(1 - \pi_{k})} \leq \Gamma, \;\; \forall \; k, k', \; \mathrm{with} \; \mathbf{x}_{kji} = \mathbf{x}_{kji'} \; \forall \; j,i,i'.
\end{equation}

If $\Gamma=1$, then  $\pi_{k} = \pi_{k^\prime}$, and the randomization distribution for $T$ is valid. If $\Gamma > 1$, then quantities such as $p$-values and point estimates are unknown but are bounded by a known interval. In a sensitivity analysis, we observe at which value of $\Gamma$ the upper bound on the $p$-value exceeds 0.05.  If the value of $\Gamma$ is large, we can be confident that it would take a large bias from a hidden confounder to reverse the conclusions of the study. The derivation for the sensitivity analysis as applied to our test statistic $T$ is in \citet{Hansen:2014}. 


Under a test of equivalence, we may be able to reject $H_{\neq}^{(\delta)}:|\tau| > \delta$ if the $p$-value from the test is low.  Rejecting this null, allows us to infer that the estimate treatment effect is not as large as $\delta$. We then apply the sensitivity analysis to understand whether this inference is sensitive to biases from nonrandom treatment assignment. In the analysis, we observe at what value of $\Gamma$ the $p$-value exceeds the conventional 0.05 threshold for each test. If this $\Gamma$ value is relatively large, we can be confident that the test of equivalence is not sensitive to hidden bias from nonrandom treatment assignment.\footnote{\citet{Hansen:2014} note that sensitivity to hidden bias may vary with the choice of weights $w_k$. To understand whether different weights lead to different sensitivities to hidden confounders, we can conduct a different sensitivity analysis for each set of weights and correct these tests using a multiple testing correction \citep{Rosenbaum:2012}. We then report a single corrected $p$-value for a value of $\Gamma$.}
 
\singlespacing
\subsection{How Much Bias Would Need to be Present to Mask a Positive Effect of Private Subsidized Schools?}
\doublespace

In this test, the null hypothesis asserts $H_{\neq}^{(\delta)}:|\tau| > \delta$ for some specified $\delta > 0$. Rejection of this null hypothesis provides evidence that the effect of attending a private voucher school on test scores is less than $\delta$. What values should we select for $\delta$?  A number of studies in the literature have found that private voucher schools increase test score achievement. The smallest effect size in the extant literature is 0.15 of a standard deviation \citep{Sapelli:2002}. However, among low income students the effects may be as large as 0.5 of a standard deviation, and \citet{Sapelli:2005} find an effect size of 0.6 standard deviations.  These results suggest a range of possible effects from 0.15 to 0.6 standard deviations.  To that end, we use three values for $\delta$ of 0.15, 0.30 and 0.6 standard deviations.  This allows us to test whether the point estimates in our study are equivalent to small, medium or large voucher effects. Thus we define three values $\delta_1$, $\delta_2$, and $\delta_3$ to correspond to these three different possible effect sizes.

We now apply the test of equivalence and sensitivity analysis to the results.  For this match, we are able to reject $H_{\neq}^{(\delta_1)}$ with $p$-value of 0.035 when $\Gamma = 1$.  Therefore, we are able to reject the null that the smaller treatment effect we observe in this design is equivalent to the smallest estimated effect in the extant literature. However, we find that when $\Gamma$ is as small as 1.09 the $p$-value for the test of equivalence is 0.053.  Thus if students differed by as much as 9 percent in the odds of being treated that could explain our inference. For a moderate effect size of 0.30 standard deviations, we can easily reject $H_{\neq}^{(\delta_2)}$ when $\Gamma =$ 1 ($p > .001$), and we find that when $\Gamma =$ 2.8 the $p$-value is 0.049.  A bias of magnitude $\Gamma =$ 2.8 means that two matched students might differ in terms of an unobserved $u_{kji}$ such that one student is almost three as like as the other to attend a private voucher school before it would alter our conclusions. Finally, for a large effect size we find that when $\Gamma =$ 9.8, the $p$-value is 0.049. Therefore, it would take a very large bias for our conclusions about a large treatment effect to be altered. 

\section{Summary and Discussion}
\label{sec:conc}

Clustered observational studies with hierarchical or multilevel data are very common in the health and social sciences.
In these settings, we have shown that the optimal matching strategy is, under the assumption that clusters have been matched optimally, first to match units and then, considering these optimal unit-level matches, to match clusters. We emphasized that this strategy explicitly uses the nested structure of the data by breaking the multilevel matching problem into simpler, smaller matching subproblems that are solved only once and that can be solved in parallel to yield an optimal global solution. We implemented this strategy using and extending cardinality matching to find the largest matched sample of pairs of treated and control units within pairs of treated and control clusters that is balanced according to specifications given by the investigator. Unlike other matching methods for multilevel data, this implementation of our matching strategy is optimal in the sense that it maximizes the size of the matched sample (or minimizes the covariate distances between matched units) and it does not require estimating the propensity score (because it directly balances covariates as specified by the investigator). As we outlined, these specifications for covariate balance are not restricted to mean balance, but extend to other forms of distributional balance such as fine balance, $x$-fine balance, and strength-$k$ balance. This multilevel matching strategy also facilitates sensitivity analyses to hidden biases due to unobserved covariates, and it readily extends to clustered observational studies with three or more levels of data.
 
The proposed multilevel matching strategy is optimal when the matching criterion is a sum of components that correspond to matched units within matched clusters, such as the total number of matched units (as in our case study), or, for example, the total sum of covariate distances between matched units.
To our knowledge, this is the first application of dynamic and integer programming ideas to observational studies.


\onehalfspacing
\bibliography{mybibliography,keele_revised2,mybibliography2}
\bibliographystyle{asa}


\end{document}